\documentclass[review]{elsarticle}

\usepackage{graphicx,url}

\usepackage[english]{babel}
\usepackage{pgf, tikz}
\usetikzlibrary{arrows, automata, fit}
\usetikzlibrary{positioning}
\usepackage{tkz-graph}
\usepackage{amsmath, amssymb}
\usepackage{booktabs}
\usepackage{todonotes}

\usepackage{amsthm}

\newtheorem{theorem}{Theorem}[section]
\newtheorem{proposition}[theorem]{Proposition}
\newtheorem{corollary}[theorem]{Corollary}

\newtheorem{lemma}[theorem]{Lemma}

\newtheorem{example}{Example}[section] 

\newcommand{\BbbR}{\mathbb{R}}
\newcommand{\cala}{{\cal A}}

\newcommand{\POA}{\ensuremath{\mathrm{PoA}}}

\newcommand{\PNE}{\ensuremath{\mathrm{NE}}}

\newcommand{\SPoA}{\ensuremath{\mathrm{SPoA}}}
\newcommand{\SPoS}{\ensuremath{\mathrm{SPoS}}}

\newcommand{\SPE}{\ensuremath{\mathrm{SPE}}}
\newcommand{\POS}{\ensuremath{\mathrm{PoS}}}

\newtheorem{defn}{Definition}[section]

\usepackage{lineno,hyperref}
\modulolinenumbers[5]

\tikzset{
    vertice/.style = {draw=black, fill=white, circle, outer  sep = 1pt, inner sep=2pt, font=\footnotesize, minimum size=15pt, solid, thin},
    weight/.style = {fill=white, circle, inner sep=1pt, font=\scriptsize},
    large vertice/.style = {vertice, minimum size=20pt, font=\scriptsize, inner sep=0},
    leaf/.style={vertice, draw=none},
    emph/.style={edge from parent/.style={dashed, red, thick,draw}},
    norm/.style={edge from parent/.style={solid, black, thin, draw}}
}

\def\lbl#1{\ifnum#1=0 edge from parent node[draw=none,swap] {$b_1$}
\else      edge from parent node[draw=none]      {$b_2$}
\fi} 

\journal{Discrete Applied Mathematics}

\begin{document}

\begin{frontmatter}

\title{Tight Bounds for the Price of Anarchy and Stability in Sequential Transportation Games}

\author[mymainaddress]{Francisco J. M. da Silva
\corref{mycorrespondingauthor}}
\cortext[mycorrespondingauthor]{Corresponding author}
\ead{francisco.silva@ic.unicamp.br}

\author[mymainaddress]{Fl{\'a}vio K. Miyazawa}
\ead{fkm@ic.unicamp.br}

\author[mymainaddress]{Ieremies V. F. Romero}
\ead{i217938@dac.unicamp.br}

\author[mymainaddress]{Rafael C. S. Schouery}
\ead{rafael@ic.unicamp.br}

\address[mymainaddress]{Institute of Computing, University of Campinas, Campinas, 13083-852, Brazil}

\begin{abstract}
    In this paper, we analyze a transportation game first introduced by Fotakis, Gourv\`es, and Monnot in 2017, where players want to be transported to a common destination as quickly as possible and, in order to achieve this goal, they have to choose one of the available buses. We introduce a sequential version of this game and provide bounds for the Sequential Price of Stability and the Sequential Price of Anarchy in both metric and non-metric instances, considering three social cost functions: the total traveled distance by all buses, the maximum distance traveled by a bus, and the sum of the distances traveled by all players (a new social cost function that we introduce). Finally, we analyze the Price of Stability and the Price of Anarchy for this new function in simultaneous transportation games.
\end{abstract}

\begin{keyword}
Transportation \sep Sequential games \sep Subgame perfect equilibrium \sep Price of Anarchy \sep Price of Stability
\end{keyword}

\end{frontmatter}


\section{Introduction}
Transportation systems are widely used every day either to go to work/school or to travel around the city. For example, if a user is at a bus station and wants to go to a shopping mall, she would have some options among all buses that travel from her location to the mall. However, ``time is money", so it is expected that (selfish) users will compete for the buses that bring them to their destination as quickly as possible.

In this work, we aim to investigate the environment where players are competing against each other for the usage of shared resources, commonly called \textit{Resource Allocation Games}. In this setting, generally, the resources either are limited or are associated with a cost, so that resource sharing is necessary or desirable. More specifically, we study a family of resource allocation games called \textit{Transportation Games}. These games were recently introduced by Fotakis et~al.~\cite{fotakis2017selfish}, and they model situations motivated by ride-sharing systems like \textit{Uber, Dial-a-ride}, or \textit{Blablacar}. These systems are also important because of their direct impact on the environment in general, as they can induce less pollutant gas emission and reduce traffic congestion~\cite{furuhata2013ridesharing}. 

A (pure) Nash Equilibrium (NE) in a game is an outcome (the actions taken by each player) where there is no player who can unilaterally decrease her cost by changing her chosen action. 
An important metric to evaluate NE outcomes was introduced by Koutsoupias and Papadimitriou~\cite{KoutsoupiasP99}, which is called \textit{Price of Anarchy} (\POA{}). It is defined as being the largest ratio among all instances of a game between the worst social cost of an equilibrium and the optimal social cost of the game. Informally speaking, the \POA{} provides us the information on how much the social cost can increase due to the selfishness of the players. For example, when the \POA{} has its value far away from~$1$, it means that the players' selfish behavior can provoke a significant increase in the social cost.

Another relevant measure of the inefficiency of equilibria is the \textit{Price of Stability} (\POS{}). It was proposed by Anshelevich et al.~\cite{anshelevich2008price} and, unlike the Price of Anarchy, it evaluates the largest ratio among all instances of a game between the best social cost of an equilibrium and the optimal social cost of the game. As a consequence, we have that $\POA{} \ge \POS{} \ge 1$. 


Fotakis et al.~\cite{fotakis2017selfish} showed tight bounds on the $\POA{}$ and the $\POS{}$ of transportation games when considering two social cost functions: $D$ (the total distance traveled by all buses, representing fuel consumption) and $E$ (an egalitarian social cost function which measures the largest cost of a user\footnote{The largest cost of a user is, in fact, the largest distance traveled by a bus.}). Those values can be seen in Table~\ref{tab_Results}.

These results are related to the \textit{simultaneous} version of this game where all players announce their actions (choose their buses) simultaneously. However, one common criticism of simultaneous games is that we do not have any sense of the sequence of the actions performed by the players, which can be unrealistic in some settings. Therefore, we aim to investigate the effect of sequential decision making in transportation games following a line of recent publications, led by Paes Leme et al.~\cite{leme2012curse}. In this setting, instead of analyzing Nash equilibria arising from simultaneous strategic moves, we analyze \textit{Subgame Perfect Equilibrium} (\SPE{}), which roughly speaking is the outcomes where players act strategically and farsighted~\cite{shoham2008multiagent}. 

The \textit{Sequential Price of Anarchy} (\SPoA{}), introduced by Paes Lemes et al.~\cite{leme2012curse}, is a tool used to measure the lack of central authority in games where players choose their actions sequentially, following some arbitrary fixed order. It compares the quality of the worst \SPE{}, considering all possible orders of the players, and the quality of an optimal social outcome. In their work, Paes Lemes~et~al.~\cite{leme2012curse} showed that the \SPoA{} has better guarantees when compared with the \POA{} in some games. For example, while the \POA{} is unbounded for the \textit{Unrelated Machine Scheduling game}, they showed that its \SPoA{} is bounded. In fact, positive results showing that the \SPoA{} presents lower values than the \POA{} has been displayed for various games (\cite{leme2012curse, de2013decentralized, de2014sequential, hassin2015sequential}). However, that is not always the case (\cite{angelucci2013sequential, bilo2012some, correa2019inefficiency}), and we show that sequential transportation games fall in this category for all social cost functions analyzed in this paper.

\subsection{Main Contributions}\label{mainContri}

First, we define a new social cost function $U$ (from \emph{utilitarian}), which is the sum of the cost of the users, and we present the sequential version of the transportation games. In particular, we give bounds for the \SPoA{} and \SPoS{} for all social cost functions considered in this paper ($U$, $E$, and $D$). Most of these bounds are tight or asymptotically tight, others are asymptotically tight if the number of buses is constant. 

In short, in Section~\ref{chap_ineff_Seq}, we first show that the value of the \SPoS{} is unbounded for non-metric instances for all the three social cost functions and, then, we proceed to show the value of the \SPoS{} and \SPoA{} for metric instances. In Table~\ref{tab_Results}, we summarize the lower (LB) and upper (UB) bounds for the inefficiency of equilibria for the metric instances of transportation games, where the columns of this table represent all measures we use to analyze the inefficiency of equilibria: Price of Anarchy~(\POA{}), Price of Stability~(\POS{}), Sequential Price of Anarchy~(\SPoA{}), and Sequential Price of Stability~(\SPoS{}).

Finally, we extend the results of Fotakis et al.~\cite{fotakis2017selfish} analyzing the inefficiency of equilibria associated with function $U$, by giving bounds on the Price of Stability~(\POS{}) and Price of Anarchy~(\POA{}) for it (see Table~\ref{tab_Results} and Section~\ref{sec:U_function}).

\begin{table}
\centering
\resizebox{\textwidth}{!}{%
\begin{tabular}{@{\extracolsep{0pt}}lcccccccc}
\toprule   
{} & \multicolumn{2}{c}{\POA{}} & \multicolumn{2}{c}{\POS{}} & \multicolumn{2}{c}{\SPoA{}} & \multicolumn{2}{c}{\SPoS{}}\\
 \cmidrule{2-3} 
 \cmidrule{4-5} 
 \cmidrule{6-7}
 \cmidrule{8-9} 
 {Function} & LB & UB & LB & UB & LB & UB & LB & UB \\ 
\midrule
$D$  & $n$ & $n$ & $n$ & $n$ & $\boldsymbol{n}$ & $\boldsymbol{n}$ & $\boldsymbol{n}$ & $\boldsymbol{n}$ \\ 
 {} & \cite{fotakis2017selfish} & \cite{fotakis2017selfish} & \cite{fotakis2017selfish} & \cite{fotakis2017selfish} & Thm.~\ref{SPoA_D} & Thm.~\ref{SPoA_D} & Prop.~\ref{prop_D_SPoS} & Thm.~\ref{SPoA_D}\\ \cmidrule{2-9}
 $E$  & $2\lceil \frac{n}{m} \rceil-1$ & $2\lceil \frac{n}{m} \rceil-1$ & $\Omega(n/m)$  &  $O(n/m)$  & $\boldsymbol{2n-1}$ & $\boldsymbol{2n-1}$ & $\boldsymbol{\lfloor\frac{n}{m}\rfloor}$ & $\boldsymbol{2n-1}$ \\ 
  & \cite{fotakis2017selfish} & \cite{fotakis2017selfish} & \cite{fotakis2017selfish} & \cite{fotakis2017selfish} & Thm.~\ref{SPOA_E} & Thm.~\ref{SPOA_E} & Prop.~\ref{prop_E_SPoS} &Thm.~\ref{SPOA_E} \\ \cmidrule{2-9}
$U$  & $\boldsymbol{2\frac{n}{m}-1}$ & $\boldsymbol{2\frac{n}{m}+1}$ & $\boldsymbol{2\frac{n}{m}-1}$ & $\boldsymbol{2\frac{n}{m}+1}$ & $\boldsymbol{2n-1}$ & $\boldsymbol{2n-1}$ & $\boldsymbol{2\frac{n}{m}-1}$ & $\boldsymbol{2n-1}$ \\  
  & Prop.~\ref{prop_util_PoS} & Thm.~\ref{functionU_upper} &  Prop.~\ref{prop_util_PoS} & Thm.~\ref{functionU_upper} & Thm.~\ref{cor_SPoA_U} & Thm.~\ref{cor_SPoA_U} & Prop.~\ref{prop_util_SPoS} & Thm.~\ref{cor_SPoA_U} \\  
\bottomrule
\end{tabular}}
\caption[A table showing all the results known for the inefficiency of equilibrium of metric instances of the transportation games for three different social cost functions]{Summary of the bounds for the inefficiency of equilibria. The values in boldface are results of this paper.} 
\label{tab_Results}
\end{table}

\section{Model and Notation}\label{def_transpGames}
In this section, we say that a complete graph $G=(V,E)$ is metric if there is an associated distance function~${d:E\rightarrow \BbbR_+}$ that assigns metric values for the edges, i.e., for every~${x, y, w \in V}$, we have that~${d(x,w) \le d(x,y) + d(y,w)}$.

\subsection{Simultaneous Transportation Game}
We first focus on simultaneous transportation games, as introduced by Fotakis et al.~\cite{fotakis2017selfish}. An instance $\Gamma$ of a simultaneous \textit{transportation game}, is a tuple $(N, M, G)$, where ${N = \{1, \ldots, n\}}$ is a set of $n$ \emph{players}; $M = \{1, \ldots, m\}$ is a set of $m \ge 2$ \emph{buses}; and $G=(V,E)$ is a complete undirected graph with a source node $s$ and a destination node $t$, where $V = N \cup \{s,t\}$; and $d:E\rightarrow \BbbR_+$ is an associated distance function. Each player $i \in N$ is placed in its corresponding vertex in $G$, and they have as a goal to be transported from their location to~$t$, while minimizing their cost. 

An \emph{outcome} in $\Gamma$ is an assignment $\sigma: N \rightarrow M$ in which every player $i \in N$ chooses one bus~$j \in M$ (the bus that will pick her up), denoted by $\sigma_i$. We call $\mathcal{P}$ the set of all outcomes and, considering an outcome $\sigma \in \mathcal{P}$, player's $i$ cost under~$\sigma$, denoted by $c_i(\sigma)$, is the distance traveled by $\sigma_i$ from location $i$ to destination~$t$. 

An outcome $\sigma \in \mathcal{P}$ is a (Nash) \emph{equilibrium} if no player can decrease her cost by changing her action, while the action of every other player remains the same. That is, $\sigma$ is an equilibrium if, for every $i \in N$ and every $\sigma' \in \mathcal{P}$ such that $\sigma_i \neq \sigma_i'$ and $\sigma_j = \sigma_j'$ for all $j \in N\setminus \{i\}$, we have that $c_i(\sigma) \leq c_i(\sigma')$.

In order to determine the routes for the buses, we consider that each bus~${j \in M}$ has an algorithm $\cala_j$, which, given $V' \subseteq V$, calculates its route which starts on node~$s$, goes through all vertices of $V'$ and finishes its route on node $t$. We consider that, as in Fotakis et al.~\cite{fotakis2017selfish},  each algorithm~$\cala_j$, for $j \in M$, is based on a permutation $\pi_j:N \rightarrow N$, given as input, such that a bus~$j$ will only visit players that have chosen bus $j$, following the order given by its permutation. That is, bus~$j$ will do ``shortcuts'' in its permutation whenever it is possible. See the example below.

\begin{example}
Consider the metric instance depicted in Figure~\ref{fig:ex_stg}, where the cost of the edges not shown are the distance of a minimum path between any pair of nodes. In this instance we have $N=\{1, \ldots, 5\}$ as the set of players and $M=\{1,2\}$ as the set of available buses. Let~$\pi_j$, for~$j \in M$, be the identity permutation, i.e. $\pi_j = (1, 2, 3, 4, 5)$, where we can interpret it as bus~$j$ following the path $s\rightarrow 1\rightarrow 2\rightarrow 3\rightarrow 4\rightarrow 5 \rightarrow t$. Notice that player $4$ will always choose a different bus than the one chosen by player $5$ because if she travels together with player~$5$, then her cost would be $7$, but if she chooses a different bus than player~$5$, then her cost would be~$3$. Let us analyze outcome $\sigma = (1, 1, 1, 2, 1)$. Observe that under~$\sigma$, buses $1$ and $2$ are going to perform shortcuts in their routes. For example, bus $2$ will follow the path $s \rightarrow 4 \rightarrow t$. Here, we have the following costs: $c_1(\sigma) = 14$, $c_2(\sigma) = 8$, $c_3(\sigma) = 5$, $c_4(\sigma) = 3$, and $c_5(\sigma) = 3$. Under $\sigma$, just player $1$ is willing to deviate and does so. Now, with this modification, we have the outcome $\sigma'=(2,1,1,2,1)$, and the improved cost of player~$1$ is $c_1(\sigma')=4$. Since under outcome $\sigma'$ no one wants to do an unilateral deviation, this is an equilibrium.

\begin{figure}
\centering
\resizebox{12cm}{!}{%
\begin{tikzpicture}    
    \node[vertice] (t) at (4,0) {$t$};
    \node[vertice] (1) at (0, -2) {$1$};
    \node[vertice] (4) at (2, -2) {$4$};
    \node[vertice] (3) at (4, -2) {$3$};
    \node[vertice] (5) at (6, -2) {$5$};
    \node[vertice] (2) at (8, -2) {$2$};
    
    \foreach \x/\y/\w in {
        1/t/3,
        2/t/3,
        3/t/3,
        4/t/3,
        5/t/3,
        1/4/1,
        4/3/2,
        3/5/2,
        5/2/1}{
        \draw (\x) to node[weight] {\w} (\y);
    }
\end{tikzpicture}
}
\caption[Example of a transportation game instance]{Metric instance with five players and their distances.} 
\label{fig:ex_stg} 
\end{figure}
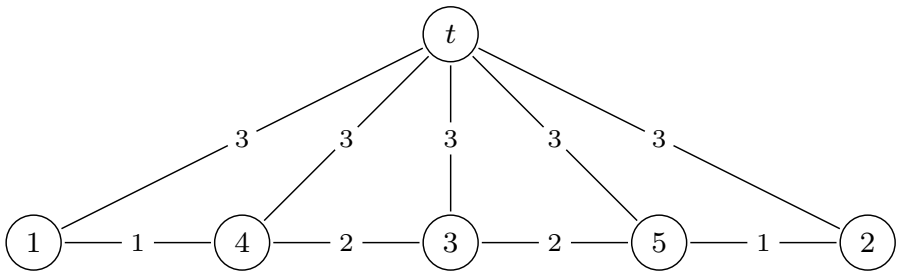
\end{example}

Given an outcome~$\sigma$, we denote the $i$-th player picked up by bus $j$ by $p_{i}^{j}$ and $|\{i: \sigma_i = j \}|$ by $n_j$.  Also, we abuse notation and consider $p_{n_j+1}^{j} = t$ for all bus~$j \in M$.

We use three different social cost functions in order to measure the inefficiency of equilibria, with two of them ($D$ and $E$) already defined and analyzed by Fotakis et al.~\cite{fotakis2017selfish}.
Since all three functions neglect the distance between $s$ and the first player picked up by a bus, we will ignore $s$ from now on.  

The first social cost function is described as \textit{Vehicle Kilometers Traveled}, which reflects the environmental impact of the game's outcome. We define
\begin{align}\label{equation:functionD}
D(\sigma)= \sum_{j=1}^{m} \sum_{i=1}^{n_j} d(p_{i}^{j}, p_{i+1}^{j}).
\end{align}
Indeed, $D(\sigma)$ represents the total distance traveled by the buses, except the distance from $s$ to the first player of each bus. 

The \textit{Egalitarian} social cost function is a classical social cost function in game theory literature which, in our context, it is the worst traveling distance of a player. It is defined as
\begin{align}\label{equation:functionE}
E(\sigma) = \max_{i \in N} c_i(\sigma).
\end{align}
Notice that this value is also the maximum distance traveled by a single bus (again, ignoring $s$).

In this paper, we introduce the analysis of another social cost function, called the \textit{Utilitarian} social cost function, which is also very common in algorithmic game-theory~\cite{roughgarden2009intrinsic, hoeksma2011price, cole2010coordination}. This function represents, in the context of transportation games, the sum of the distances travelled by all players, and it is defined as
\begin{align}\label{equation:functionU}
U(\sigma) = \sum_{i \in N} c_i(\sigma) = \sum_{j \in M} \sum_{i \in N \,\vert\, \sigma_i = j} c_i(\sigma).
\end{align}
Thus, $U(\sigma)/n$ can be seen as the average distance a player will travel to reach her final destination. 

\begin{defn}\label{defn:PoA} Given a transportation game $\Gamma$ and a social cost function ${f: \mathcal{P} \to \mathbb{R}}$, let \PNE{}($\Gamma$) be the set of Nash equilibria of $\Gamma$ and let $\sigma^* \in \mathcal{P}$ be an outcome that minimizes $f(\sigma^*)$. The \textit{Price of Anarchy} (\POA{}) of $\Gamma$ for function $f$ is defined as 
\begin{equation}
\POA{}(f, \Gamma) = \max_{\sigma \in \PNE{}(\Gamma)} \frac{f(\sigma)}{f(\sigma^{*})},
\end{equation}
while the \textit{Price of Stability} (\POS{}) of $\Gamma$ for function $f$ is defined as 
\begin{equation}
\POS{}(f, \Gamma) = \min_{\sigma \in \PNE{}(\Gamma)} \frac{f(\sigma)}{f(\sigma^{*})}.
\end{equation}
\end{defn}

The Price of Anarchy and the Price of Stability of a class~$\mathcal{G}$ of transportation games are defined by ${\POA{}(f, \mathcal{G}) = \sup_{\Gamma \in \mathcal{G}} \POA{}(f, \Gamma)}$ and ${\POS{}(f,\mathcal{G}) = \sup_{\Gamma \in \mathcal{G}} \POS{}(f, \Gamma)}$, respectively. We will use $\POA{}(f)$ and $\POS{}(f)$ when the class of transportation games $\mathcal{G}$ is clear from the context.

Fotakis et al.~\cite{fotakis2017selfish} showed that a NE always exists when all the buses have the same permutation or in metric instances with two buses, but they also proved that not all metric instances of this game possess a \PNE{}. Because of this, they analyzed the inefficiency of equilibria of only the instances of the game that do have a \PNE{}. We also make this assumption when analyzing the~$\POA{}(U)$ and the~$\POS{}(U)$ in Section~\ref{sec:U_function}.

\subsection{Sequential Transportation Games}\label{def_sec:seq_model}
In the sequential version of the transportation game, players still choose a bus $j \in M$ but, instead of announcing their actions simultaneously, they choose their actions following an arbitrary predefined order. For simplicity, we consider this order to be $(1,2,\ldots,n)$ as one can always change buses' permutation accordingly. Therefore, player $i$ has to choose an action~$\sigma_i$ knowing only the actions taken by the players $1,2,\ldots, i-1$. However, because we are considering a game with full information, when choosing an action, players will indeed take into consideration their successors' behavior, so they will be able to fully anticipate their actions considering that they will behave selfishly.

Because of the game's sequentiality, for each set of possible actions ${\sigma_{<i} = (\sigma_1, \sigma_2, \ldots, \sigma_{i-1})}$ chosen by the predecessors of player~$i$, player~$i$ needs to specify an action $\lambda_i(\sigma_{<i})$ to deal with all of those actions, that is, specify which bus should be chosen if player $1$ to $i-1$ chooses actions $\sigma_{<i}$. We call~$\lambda_i$ a \textit{strategy}, since it defines an action for every possible predecessors' choices. Also, we use $\lambda = (\lambda_1,\lambda_2, \ldots, \lambda_n)$ to refer to the strategies' choice of all players, which we call a \emph{strategy profile} and we denote the set of all strategy profiles by~$\mathcal{S}$. 

Finally, the \emph{outcome of} $\lambda$ is an outcome $\sigma = (\sigma_1, \sigma_2,\dots, \sigma_n)$ such that $\sigma_i = \lambda_i(\sigma_{<i})$ for all $i \in N$, where, again, ${\sigma_{<i} = (\sigma_1, \sigma_2, \ldots, \sigma_{i-1})}$. That is, the outcome of $\lambda$ is such that $\sigma_i$ is the action chosen by player $i$ (according to $\lambda$) when player $j$ chooses $\sigma_j$ for $1 \leq j < i$.

Following the definition given by Shoham and Leyton-Brown~\cite{shoham2008multiagent}, we say that a strategy profile $\lambda$ is a \textit{Subgame Perfect Equilibrium} (\SPE{}) if it induces a Nash Equilibrium in any subgame (the restriction of the game to some $\sigma_{<i}$). That is,~$\lambda$ is a \SPE{} if for all players $i$ and all~$\sigma_{<i}$, $i$ cannot decrease her cost by changing to another strategy different from $\lambda(\sigma_{<i})$ in the \textit{subgame} where~$\sigma_{<i}$ is fixed for her predecessors (players $1, 2, \ldots, i-1$) and $\lambda_{i+1},\lambda_{i+2},\ldots, \lambda_{n}$ are the strategies chosen by her successor (players $i+1, i+2, \ldots, n$).

Furthermore, we can express the sequential transportation game in its extensive form by considering a rooted $m$-ary tree, wherein each level $i$ ($1 \le i \le n)$, player~$i$ is responsible for taking decisions of all nodes in her level. In other words, she has to choose one of the $m$ buses knowing only the decisions of her~$i-1$ predecessors. On the leaves (terminal nodes) is the information about the cost of each player according to that outcome (which can be seen as a path of choices from the root to that leaf). Since we are dealing with a game with full information, it always has at least one \SPE{}, which can be computed by backward induction, using the well-known Zermelo's algorithm~\cite{zermelo1913anwendung}.
 
Regarding the inefficiency of equilibria, Paes Leme et al.~\cite{leme2012curse} introduced the notion of the Sequential Price of Anarchy (\SPoA{}) as a tool to measure the quality of \SPE{}s.

Consider $f$ as one of the previously defined social cost function ($U$, $D$ or $E$). Notice that $f: \mathcal{P} \to \mathbb{R}$ and, thus, it measures the social cost of an outcome of the sequential game. For some strategy profile $\lambda \in \mathcal{S}$, we will abuse notation and write $f(\lambda)$ to denote the social cost of the outcome of $\lambda$ according to $f$.

\begin{defn}\label{defn:SPoA} Given a sequential transportation game $\Gamma$ and a social function ${f: \mathcal{P} \to \mathbb{R}}$, let \SPE{}($\Gamma$) be the set of all Subgame Perfect Equilibria of $\Gamma$ and $\sigma^*$ be an outcome that minimizes $f(\sigma^*)$. The \textit{Sequential Price of Anarchy} (\SPoA{}) of $\Gamma$ for function $f$ is defined as
\begin{equation}
\SPoA{}(f, \Gamma) = \max_{\lambda \in \SPE{}(\Gamma)} \frac{f(\lambda)}{f(\sigma^{*})},
\end{equation}
while the \textit{Sequential Price of Stability} (\SPoS{}) of $\Gamma$ for function $f$ is defined as 
\begin{equation}
\SPoS{}(f, \Gamma) = \min_{\lambda \in \SPE{}(\Gamma)} \frac{f(\lambda)}{f(\sigma^{*})}.
\end{equation}
\end{defn}

The Sequential Price of Anarchy and the Sequential Price of Stability of a class of sequential transportation games $\mathcal{G}$ are defined by ${\SPoA{}(f,\mathcal{G}) = \sup_{\Gamma \in \mathcal{G}} \SPoA{}(f, \Gamma)}$ and ${\SPoS{}(f,\mathcal{G}) = \sup_{\Gamma \in \mathcal{G}} \SPoS{}(f, \Gamma)}$, respectively. 
Again, we will use $\SPoA{}(f)$ and $\SPoS{}(f)$ when the class of sequential transportation games $\mathcal{G}$ is clear from the context. Next, we show an example of an instance of a sequential transportation game.

\begin{example}
Consider Figures~\ref{fig:ex_spos1} and~\ref{fig:ex_spos2}. In this game, there are~${n=4}$ players and~${m=2}$ buses with ${\pi_{1} = \pi_{2} = (1,2,4,3)}$. In Figure~\ref{fig:ex_spos2}, we have the game depicted in its extensive form where it is also indicated the \SPE{}~${\lambda = (1, (1,2), (2,1,1,2),(2,1,2,1, 2,1,2,1))}$ with outcome $(1, 1, 2, 1)$, where actions of a player~$i$ is ordered from left to right considering the nodes of level~$i$. Also, the leaves represent the cost associated with each of players $1$, $2$, $3$, and~$4$, respectively. Under the Egalitarian social function, we can see that, for this instance,~${\SPoA{}(E) \ge 5/3}$ as an optimal outcome values $3$ (when players $1$ and $2$ are together in one bus and $3$ and~$4$ are together in the other bus).
\end{example}

    \begin{figure}[h]
    
    \centering
    
    \begin{tikzpicture}
    \node[vertice] (t) at (0,0) {$t$};
    \node[vertice] (1) at (2,0) {$1$};
    \node[vertice] (2) at (4,0) {$2$};
    \node[vertice] (3) at (-4,0) {$3$};
    \node[vertice] (4) at (-2,0) {$4$};

    \draw (3) to node[weight] {$1$} (4);
    \draw (4) to node[weight] {$1$} (t);
    \draw (1) to node[weight] {$1$} (t);
    \draw (2) to node[weight] {$1$} (1);
\end{tikzpicture}
    \caption{Metric instance with four players on a line.}
    \label{fig:ex_spos1}
\end{figure}
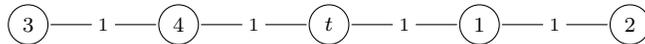

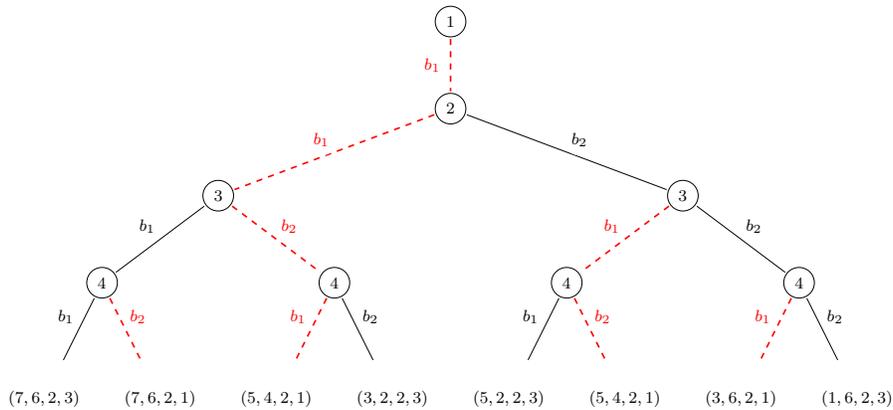
\begin{figure}[h]
    \centering
    \resizebox{12cm}{!}{%
    \begin{tikzpicture}[auto,
    every node/.style={vertice},    
    level 1/.style={sibling distance=8cm},
    level 2/.style={sibling distance=8cm},
    level 3/.style={sibling distance=4cm},
    level 4/.style={sibling distance=2cm,level distance=20mm}
    ]

    \node (1) {$1$} 
    child[emph] {node (2) {$2$}
        child[emph] {node (3) {$3$}
            child[norm] {node {$4$} 
                child {node[leaf] {$(7,6,2,3)$}   \lbl{0}}
                child[emph] {node[leaf] {$(7,6,2,1)$}   \lbl{1}}  
                \lbl{0}}
            child[emph] {node {$4$}   
                child[emph] {node[leaf] {$(5,4,2,1)$}   \lbl{0}}
                child[norm] {node[leaf] {$(3,2,2,3)$}   \lbl{1}} 
                \lbl{1}}
              \lbl{0}}
        child[norm] {node {$3$} 
            child[emph] {node {$4$}   
                child[norm] {node[leaf] {$(5,2,2,3)$}   \lbl{0}}
                child[emph] {node[leaf] {$(5,4,2,1)$}   \lbl{1}}
                \lbl{0}}
            child[norm] {node {$4$}   
                child[emph] {node[leaf] {$(3,6,2,1)$}   \lbl{0}}
                child {node[leaf] {$(1,6,2,3)$}   \lbl{1}}
                \lbl{1}}  
            \lbl{1}}
        \lbl{0}};
\end{tikzpicture}
}
    \caption{A tree representing a sequential transportation game. Because the permutations are equal, the tree is symmetric and we draw here just one half of it, where player~$1$ chooses the first bus. Dashed red edges indicate the action chosen at each node for \SPE{}~${\lambda = (1, (1,2), (2,1,1,2),(2,1,2,1, 2,1,2,1))}$. In the leaves, we have the cost of each player in the associated outcome.}
    \label{fig:ex_spos2}
\end{figure}

\section{Inefficiency of Equilibria in the Sequential Transportation Games}\label{chap_ineff_Seq}
Here, we analyze the Sequential Price of Anarchy and the Sequential Price of Stability for the sequential transportation games by considering the social cost functions $U$, $E$, and $D$. We begin by showing that the \SPoS{}, and thus the~\SPoA{}, is unbounded for all these social cost functions when dealing with non-metric instances. Then, we discuss our results about the value of the \SPoS{} and \SPoA{}  for metric cases of the sequential transportation games.
\subsection{Non-metric Instances}
As one could imagine, following the results of Fotakis et al.~\cite{fotakis2017selfish} stating that the \POS{} is unbounded for non-metric instances of the transportation game in its simultaneous version, we show that the \SPoS{} is also unbounded for social cost functions $U$, $E$, and $D$ when dealing with non-metric instances, even if all buses' permutations are equal.

\begin{figure}
\centering 
\begin{tikzpicture}
    \node[vertice] (t) at (0,0) {$t$};
    \node[vertice] (1) at (90:2) {$1$};
    \node[vertice] (2) at (210:2) {$2$};
    \node[vertice] (3) at (-30:2) {$3$};

    \draw (1) to node[weight] {$X$} (t);
    \draw (1) to node[weight] {$X$} (3);
    \draw (1) to node[weight] {$0$} (2);
    \draw (2) to node[weight] {$0$} (3);
    \draw (2) to node[weight] {$0$} (t);
    \draw (3) to node[weight] {$1$} (t);
\end{tikzpicture}
\caption[A metric instance of the sequential transportation game which is used to show a lower bound on the Sequential Price of Anarchy for social function $D$]{Non-metric instance with $n=3$ players and $m=2$ buses, where $X$ is a positive number.}
\label{fig:SPoS_Unb}
\end{figure}
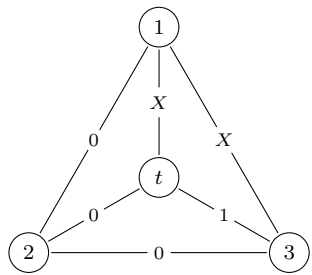
\begin{proposition}\label{SPoA_unb}
For non-metric sequential transportation games, the \SPoS{}$(U)$, \SPoS{}$(E)$, and \SPoS{}$(D)$ are unbounded even restricted to all buses having the same permutation.
\end{proposition}
\begin{proof}
Consider the instance $\Gamma$ showed in Figure~\ref{fig:SPoS_Unb} with $n=3$, $m=2$, and $\pi$ being equal to the permutation ${\pi_{b_1}=\pi_{b_2}=(1,3,2)}$. In Figure~\ref{fig:SPoS_Unb_Ext}, we can see the game in its extensive form where the leaves represent the cost of players $1$, $2$, and $3$ respectively. Observe that, in any \SPE{}, player $3$ will always choose the same bus chosen by player $2$, so both of them will get cost $0$. Then, player $1$ cannot escape from getting cost $X$.

Now, in any optimal outcome $\sigma^*$ for any social cost function $f \in \{U,E,D\}$, player $3$ is alone in one bus while player~$1$ and~$2$ are on the other bus, so $f(\sigma^*) = 1$ for every $f \in \{U,E,D\}$. Therefore, when $X \to\infty$, the \SPoS{}$(f, \Gamma)$ tends to $\infty$ and the result follows. 
\end{proof}

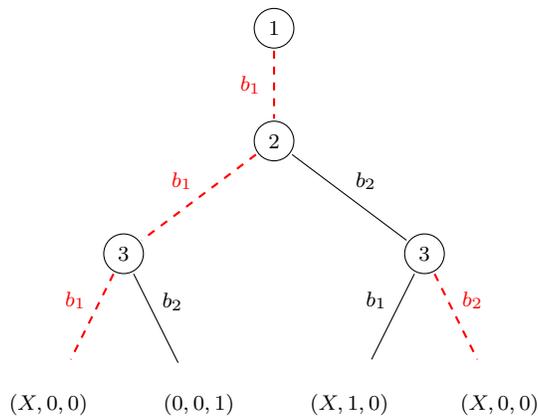
\begin{figure}
\centering 
\begin{tikzpicture}[auto,
    every node/.style={vertice},    
    level 1/.style={sibling distance=8cm},
    level 2/.style={sibling distance=4cm},
    level 3/.style={sibling distance=2cm,level distance=20mm}
    ]

    \node (1) {$1$} 
    child[emph] {node (2) {$2$}
        child[emph] {node (3) {$3$}
            child[emph] {node[leaf] {$(X, 0, 0)$}   \lbl{0}}
            child[norm] {node[leaf] {$(0, 0, 1)$}   \lbl{1}}  
            \lbl{0}}
        child[norm] {node {$3$} 
                child[norm] {node[leaf] {$(X, 1, 0)$}   \lbl{0}}
                child[emph] {node[leaf] {$(X, 0, 0)$}   \lbl{1}}
                \lbl{1}}
        \lbl{0}};
\end{tikzpicture}
\caption[A metric instance of the sequential transportation game which is used to show a lower bound on the Sequential Price of Anarchy for social function $D$]{Representation of the \SPE{} $(b_1, (b_1, b_1), (b_1,b_2,b_1,b_2))$. Because the permutations are equal, the tree is symmetric and we draw here just one half of it, where player~$1$ chooses the first bus. Dashed red edges indicate the action chosen at each node.}
\label{fig:SPoS_Unb_Ext}
\end{figure}

\subsection{Function D with Metric Instances}
Next, we show the value of the \SPoS{} and \SPoA{} for metric instances in relation with social function~$D$. It turns out that this value is $n$, which is equal to the value of the $\POA{}(D)$ and the $\POS{}(D)$ for its simultaneous version~\cite{fotakis2017selfish}. First, we show that $\SPoS{}(D) \ge n$ and, then, we proceed by showing that $\SPoA{}(D) = n$.
\begin{proposition}\label{prop_D_SPoS}
For metric transportation games with $n$ players, ${\SPoS{}(D) \ge n}$, even restricted to all buses having the same permutation.
\end{proposition}
\begin{proof}
For this lower bound, consider a game $(N,M,G)$ where ${|N|= |M| = n}$, and the graph is the one showed in Figure~\ref{fig:SPoA_D}. Let $\pi_j$, for~${j \in M}$, be the permutation $(n, n-1, \ldots, 2,1)$. It is possible to see that, as long as~${\varepsilon < \frac{1}{n-1}}$, there is an outcome $\sigma^*$ where all players are on the same bus, so ${D(\sigma^*) = 1 + (n-1)\varepsilon}$.

We will show, by induction on the player's index, that in the outcome of any \SPE{}~$\lambda$ of this instance, each bus is being used by a single player and, therefore, we get that the $D(\lambda) = n$.

For player $1$, her cost is the same no matter what is the choice of the other players. Now, suppose that all players $1, \ldots, i-1$ choose a different bus in the outcome of $\lambda$. For player $i$, any of the $n - i + 1$ buses without any of $i$'s predecessors cost $1$ and every other bus have cost $1 + \varepsilon$, no matter what is the choice of players $i + 1, \dots, n$. Thus, since $\lambda$ is a \SPE{}, $i$ chooses a bus different than the buses chosen by players $1, \ldots, i-1$. Therefore, in any \SPE{} $\lambda$, $D(\lambda)$ is~$n$, and hence, ${\SPoS{}(D) \ge \frac{n}{1+(n-1)\varepsilon}}$ for any $\varepsilon > 0$, from where the result follows.
\end{proof}

\begin{figure}
\centering 
\begin{tikzpicture}
    
    \node[large vertice] (1) at (162-18:1.5) {$1$};

    \node[large vertice] (2) at (234-18:1.5) {$2$};

    \node[large vertice, draw=none] (dot) at (306-18:1.5) {$\dots$};

    \node[large vertice] (n-1) at (18-18:1.5) {$n-1$};
    
    \node[large vertice] (n) at (90-18:1.5) {$n$};
    
    \node[large vertice] (t) at (7, 0) {$t$};

    \foreach \x/\y in {1/2, 1/dot, 1/n-1, 1/n, 2/dot, 2/n-1, 2/n, dot/n-1, dot/n, n-1/n} {
        \draw (\x) to node[weight] {$\varepsilon$} (\y);
    }

    \draw (1) to[bend left=45] node[weight] {$1$} (t);
    \draw (2) to[bend right=45] node[weight] {$1$} (t);
    \draw (dot) to node[weight] {$1$} (t);
    \draw (n-1) to node[weight] {$1$} (t);
    \draw (n) to node[weight] {$1$} (t);

\end{tikzpicture}
\caption[A metric instance of the sequential transportation game which is used to show a lower bound on the Sequential Price of Anarchy for social function $D$]{Graph $G$ where $d(s,u) = d(u,t)=1$ for all $u \in N$ and $d(u,v)=\varepsilon$ for all $u,v \in N$.}
\label{fig:SPoA_D}
\end{figure}
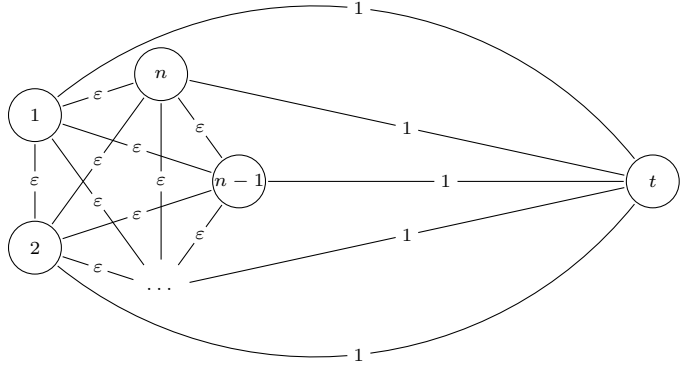

\begin{corollary}\label{SPoA_D}
    For metric sequential transportation games with $n$ players, $\SPoS{}(D)=\SPoA{}(D) = n$.
\end{corollary}
\begin{proof}
For the upper bound, we have from Fotakis et al.~\cite{fotakis2017selfish} that $D(\sigma) \le nD(\sigma^*)$ for any outcome $\sigma$ and optimal outcome $\sigma^*$ and, thus, ${\SPoS{}(D) \le \SPoA{}(D) \le n}$. Now, we get the lower bound from Proposition~\ref{prop_D_SPoS}, and the proof is done.
\end{proof}

\subsection{Function E with Metric Instances}
We begin by showing a lower bound on the $\SPoS{}(E)$, and for this, we show in Example~\ref{graphGroup} a family of instances proposed by Fotakis et al.~\cite{fotakis2017selfish} for proving that ${\POS{}(D) \ge n}$.

\begin{example}[Fotakis et al.~\cite{fotakis2017selfish}]\label{graphGroup}
We construct the following metric instance considering that $k$ is a positive integer. The set of players is $N = L \cup R$ where~${L = \{l_{i,j}\colon  1 \le i \le k, 1 \le j \le m\}}$ (the left players) and ${R = \{r_{i,j}\colon  1 \le i \le k, 1 \le j \le m\}}$ (the right players). Therefore, we have ${n=2km}$ players. Next, we decompose $L$ and~$R$ into $k$ levels, where ${L_i=\{l_{i,j}\colon 1 \le j \le m\}}$ and ${R_i=\{r_{i,j}\colon  1 \le j \le m \}}$. Also, consider graph~$G$ as depicted in Figure~\ref{fig:graphGroupInstance}.

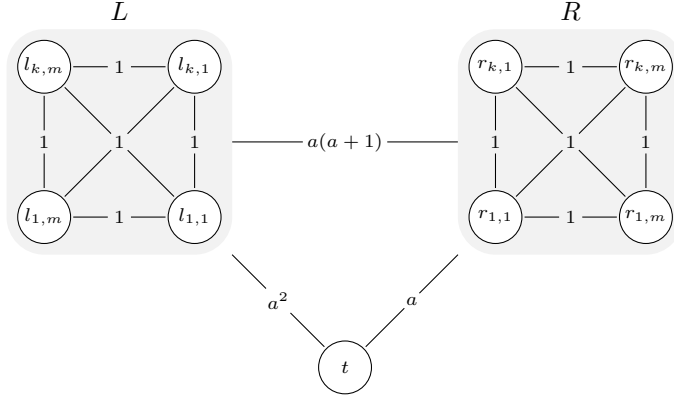
\begin{figure}
\centering
\begin{tikzpicture}

\tikzset{mysquare/.style={large vertice, rectangle, minimum width=3cm, minimum height=3cm, outer sep=0, font=\normalsize, rounded corners=0.5cm, fill=gray!10, draw=none},
gweight/.style={weight, fill=gray!10}
}

\node[mysquare, label=90:{$L$}] (L) at (-3, 3) {};
\node[mysquare, label=90:{$R$}] (R) at (3, 3) {};
\node[large vertice] (t) at (0,0) {$t$};
\draw (t) to node[weight] {$a^2$} (L);
\draw (t) to node[weight] {$a$} (R);
\draw (L) to node[weight] {$a(a + 1)$} (R);

\node[large vertice] (l11) at (-4, 4) {$l_{k,m}$};
\node[large vertice] (l12) at (-2, 4) {$l_{k,1}$};
\node[large vertice] (l21) at (-4, 2) {$l_{1,m}$};
\node[large vertice] (l22) at (-2, 2) {$l_{1,1}$};

\draw (l11) to node[gweight] {$1$} (l12);
\draw (l11) to node[gweight] {$1$} (l21);
\draw (l11) to node[gweight] {$1$} (l22);
\draw (l21) to node[gweight] {$1$} (l12);
\draw (l21) to node[gweight] {$1$} (l22);
\draw (l12) to node[gweight] {$1$} (l22);

\node[large vertice] (r11) at (4, 4) {$r_{k,m}$};
\node[large vertice] (r12) at (2, 4) {$r_{k,1}$};
\node[large vertice] (r21) at (4, 2) {$r_{1,m}$};
\node[large vertice] (r22) at (2, 2) {$r_{1,1}$};

\draw (r11) to node[gweight] {$1$} (r12);
\draw (r11) to node[gweight] {$1$} (r21);
\draw (r11) to node[gweight] {$1$} (r22);
\draw (r21) to node[gweight] {$1$} (r12);
\draw (r21) to node[gweight] {$1$} (r22);
\draw (r12) to node[gweight] {$1$} (r22);

\end{tikzpicture}

\caption[A metric instance of the transportation game which is used to give an upper bound on the Price of Stability of social cost functions $D$ and $E$]{Graph $G$ with the following distances (where $a$ is some positive number): $d(u,v) = 1$ if $u,v \in L$ or $u,v \in R$; $d(v,t) = a$ if $v \in R$; $d(v,t) = a^2$ if $v \in L$; and ${d(u,v) = a(a+1)}$ if $u \in L$ and $v \in R$. }
\label{fig:graphGroupInstance}
\end{figure}

The buses' permutations are all equal to a permutation $\pi$ defined as follows: 
every player from level $i$ appears after every player from level $j$ with $j > i$;
every player from $L_i$ appears after every player from $R_i$ for every level $i$;
and every player $l_{i,j} \in L_i$ (resp. $r_{i,j} \in R_i$) appears after every player $l_{i, k}$ (resp. $r_{i,k} \in R_i$) with~$k > j$.
\end{example}

\begin{proposition}\label{prop_E_SPoS}
For metric transportation games with $n$ players and $m$ buses, ${\SPoS{}(E) \ge \lfloor\frac{n}{m}}\rfloor$, even restricted to all buses having the same permutation.
\end{proposition}
\begin{proof}
To show this lower bound, we will use the instances described in Example~\ref{graphGroup}. We consider that the players are labeled such that the order $\pi$ presented in Example~\ref{graphGroup} is $(n, n-1, \ldots, 1)$ (e.g., player $1$ is $l_{1,1}$ and player $n$ is~$r_{k, m}$).

It is possible to see that there is an outcome $\sigma^*$ where all buses are being used in the following way. If $m$ is even, each bus picks up $2k$ players from only one group ($L$ or $R$), two per level. Hence, $E(\sigma^*) = a^2 + 2k - 1$. If $m$ is odd, we use one bus to pick up one player from every level of $L$, one bus to pick up three players from every level of $R$, and the other buses described in $m$ even case. Thus, considering that $a + 3k - 1 < a^2 + 2k - 1$ (which is true for sufficient large values of $a$), we obtain, again, that~$E(\sigma^*) = a^2 + 2k - 1$.

Now, we will show that in any outcome of a $\SPE{}$ $\lambda$, each bus contains exactly one player of each level of $L$ and one player of each level of $R$. 

Notice that Player~$1$~($l_{1,1}$) can choose any bus as she comes last in~$\pi$ and obtain the same cost for any choice. Now, for player $p$, suppose that every player $j < p$ chooses a bus such that, every bus has at most one player of each group ($L$ or~$R$) and each level when considering exclusively players $1, \ldots, p - 1$. Observe that a player $p$ cares only about the choices of her predecessors, i.e. players $1,2, \ldots, p - 1$, since they are the only ones to come after her in $\pi$. 

Suppose player $p$ is player $l_{i,j}$. Notice that, in this case, every bus picks up one player of each group and each level up to level $i - 1$, but exactly $j - 1$ of them also picks up an additional player on $l_{i, j'}$ with $j' < j$. Thus, $l_{i,j}$ chooses a bus with exactly one player of each group and each level up to level $i - 1$ and no player $l_{i, j'}$ with $j' < j$ since its cheaper (by one unit). A similar argument can be done when $p$ is player $r_{i,j}$. 

Therefore, we have that $E(\lambda) = (2k-1)(a^2+a) + a^2$ and, when $a \to \infty$, $\frac{E(\lambda)}{E(\lambda^*)}$ tends to $2k = \frac{n}{m}$. Notice that a similar proof can be made for the case where $m$ does not divide $n$ by inserting additional players located at $t$ and at the end of~$\pi$, from where the results follows. 
\end{proof}

Next, in contrast with the results on the value of \POA{} related to function $E$ presented by Fotakis et al.~\cite{fotakis2017selfish} ($\POA{}(E) = 2\lceil \frac{n}{m} \rceil - 1$ for $n > m$ and $\POA{}(E) = 1$ if $n \le m$), the next theorem shows that the value of the \SPoA{} is worse than in its simultaneous version for function $E$ even when~${n=m}$.

\begin{theorem}\label{SPOA_E}
    For metric transportation games with $n$ players and $n$ buses, ${\SPoA{}(E) = 2n-1}$.
\end{theorem}
\begin{proof}
Let $\sigma^*$ be an optimal outcome. The maximum value an outcome can achieve is when all players choose to travel on a single bus, and we argue that it is an upper bound on the $\SPoA{}(E)$ of ${(2n-1)E(\sigma^*)}$. This value comes from Fotakis et al.~\cite{fotakis2017selfish}, in which they proved that ${d(u,v) \le 2E(\sigma^*)}$ for all pairs $u,v \in N$ and $d(u,t) \le E(\sigma^*)$ for all $u \in N$. Then, since all players are on a single bus, we have that the path which will be used by it has only one edge directly connected to $t$ and the remaining $n-1$ edges are used to pick up all players. Hence, this path values at most ${(n-1)2E(\sigma^*) + E(\sigma^*) = (2n-1)E(\sigma^*)}$.

Now, for the lower bound, we provide a family of instances containing one \SPE{} with value that matches the given upper bound. These instances are given by $(N,M,G)$ where $|N| = |M| = n$, and the graph depicted in Figure~\ref{fig:SPoA_D} with~${\varepsilon = 2}$. Let~$\pi_j$, for $j \in M$, be the identity permutation, i.e., $\pi_j = (1, \ldots, n)$. Observe that an optimal outcome $\sigma^*$ is the one where each bus is being used by a single player, and hence we have that $E(\sigma^*) = 1$. 

We will show, by backward induction on the player's index, that there exists at least one \SPE{} $\lambda$ where, in its outcome, all players choose the same bus, and therefore we get that~${\SPoA{}(E) \ge 2n-1}$. For player $n$, since she is the last player to be picked up in all buses, her cost will always be $1$ despite the bus she chooses to travel in. Then, for any $\sigma_{<n}$, she can set $\lambda_n(\sigma_{<n}) = \sigma_{n-1}$ (that is always choose the same  bus chosen by player~${n-1}$) and still obtain cost $1$.

Now, suppose the claim is valid for all players $k+1, \ldots, n$, and consider player~$k$. Here, observe that the actions taken by her predecessors do not influence her cost because all of them are being picked up before her according to permutations $\pi$. She has $m$ options of buses which will all give her a cost of ${2n-2k+1}$ since, by the induction hypothesis, players $k+1, \ldots, n$ will choose the same bus as player $k$, and, therefore, player~$k$ cannot avoid traveling with all of them. Thus, for any $\sigma_{<k}$, she can set $\lambda_k(\sigma_{<k}) = \sigma_{k-1}$ with a cost of ${2n-2k+1}$, since her cost will not be influenced by the decisions of her predecessors and she cannot avoid her successors. To have a better insight, take a look at the decision tree for an instance with~$3$ players, depicted in Figure~\ref{fig:ex_spoa2}.   

This will lead to an outcome where all players are choosing bus~$\sigma_1$ at each choice node of the extensive form of this game. Thus, 
$E(\lambda) = 2n - 1$, from where the result follows.
\end{proof}

\begin{figure}
    \centering
    \resizebox{\textwidth}{!}{%
    \begin{tikzpicture}[
    scale = 0.6,
    every node/.style={vertice, font=\large, scale=0.6},    
    level 1/.style={sibling distance=8cm},
    level 2/.style={sibling distance=2.5cm},
    level 3/.style={sibling distance=1cm},
    level 4/.style={sibling distance=2cm,level distance=20mm},
    ]
    \node {1}
child[emph] {
    node {2}
    child[emph] {
        node {3}
            child[emph] {
                node[leaf] {}
            }
            child[norm] {
                node[leaf] {}
            }    
            child[norm] {
                node[leaf] {}
            }
        }
    child[norm] {
        node {3}
        child[norm] {
                node[leaf] {}
            }
            child[emph] {
                node[leaf] {}
            }    
            child[norm] {
                node[leaf] {}
            }
    }    
    child[norm] {
        node {3}
        child[norm] {
                node[leaf] {}
            }
            child[norm] {
                node[leaf] {}
            }    
            child[emph] {
                node[leaf] {}
            }
    }
}
child[norm] {
    node {2}
    child[norm] {
        node {3}
        child[emph] {
                node[leaf] {}
            }
            child[norm] {
                node[leaf] {}
            }    
            child[norm] {
                node[leaf] {}
            }
    }
    child[emph] {
        node {3}
        child[norm] {
                node[leaf] {}
            }
            child[emph] {
                node[leaf] {}
            }    
            child[norm] {
                node[leaf] {}
            }
    }    
    child[norm] {
        node {3}
        child[norm] {
                node[leaf] {}
            }
            child[norm] {
                node[leaf] {}
            }    
            child[emph] {
                node[leaf] {}
            }
    }    
}
child[norm]{
    node {2}
    child[norm] {
        node {3}
        child[emph] {
                node[leaf] {}
            }
            child[norm] {
                node[leaf] {}
            }    
            child[norm] {
                node[leaf] {}
            }
    }
    child[norm] {
        node {3}
        child[norm] {
                node[leaf] {}
            }
            child[emph] {
                node[leaf] {}
            }    
            child[norm] {
                node[leaf] {}
            }
    }    
    child[emph] {
        node {3}
        child[norm] {
                node[leaf] {}
            }
            child[norm] {
                node[leaf] {}
            }    
            child[emph] {
                node[leaf] {}
            }
    }    
};  
\end{tikzpicture}
    }
    \caption[A Subgame Perfect Equilibrium of a transportation game used to show a lower bound on the Sequential Price of Anarchy for social function $E$]{Decision tree for the game with three players. Consider that each of the three edges of each node is labeled as follows, from left to right: bus $1$, $2$, and $3$. Then, the dashed red edges indicate the action chosen at each node.}
    \label{fig:ex_spoa2}
\end{figure}
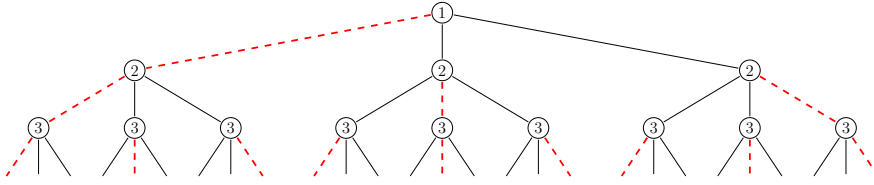

\subsection{Function U with Metric Instances}
We begin by showing a lower bound for $\SPoS{}(U)$.


\begin{proposition}\label{prop_util_SPoS}
For metric transportation games with $n$ players and $m$ buses, ${\SPoS{}(U) \ge 2\frac{n}{m}-1}$, even restricted to all buses having the same permutation.
\end{proposition}
\begin{proof}
For this lower bound, consider an instance $(N,M,G)$ where $|N|=n$, $|M| = m$, and the graph is the one showed in Figure~\ref{fig:pos_U}. For~${j \in M}$, let ${\pi_j = (1, \ldots, n-m, n, n-1, \ldots, n-m+1)}$. It is possible to see that there is an outcome $\sigma^*$ where players $\{n-m+1,\ldots, n\}$ are on bus~$1$ and the remaining players are on bus~$2$, and thus ${U(\sigma^*)=m+\frac{m(m-1)}{2}\varepsilon}$.  

Now, in any $\SPE{}$ $\lambda$, player $n$ will always choose a different bus chosen by any of the players~${\{n-m+1, \ldots, n-1\}}$ since she is the last one to enter in the game. She can do this because there are $m$ buses and she is in the last~$m$ positions of all buses' permutations, and therefore she will guarantee herself a cost of $1$. Again, the same argument can be done backwardly for players~${\{n-m+1, \ldots, n-1\}}$, which will also guarantee to all of them a cost of~$1$. Finally, the remaining players will get cost $2$ independently of their buses' choice, so $U(\lambda) = m+2(n-m) = 2n - m$. Therefore, \[{\SPoS{}(U) \ge \frac{2n - m}{m+\frac{m(m-1)}{2}\varepsilon}},\] and the results follows when $\varepsilon$ goes to zero.
\end{proof}

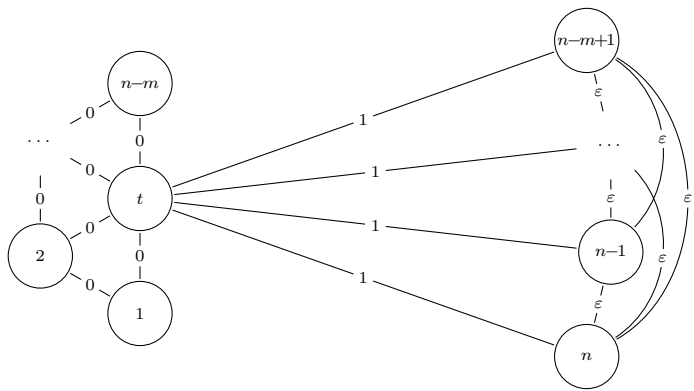
\begin{figure}
\centering 
\begin{tikzpicture}[modified large vertice/.style={large vertice, minimum size=27pt}, scale=0.9, every node/.style={scale=0.9}]

    \node[modified large vertice] (t) at (0,0) {$t$};
    \node[modified large vertice] (1) at (270:1.7) {$1$};
    \node[modified large vertice] (2) at (150+60:1.7) {$2$};
    \node[modified large vertice, draw=none] (dot) at (150:1.7) {$\dots$};
    \node[modified large vertice] (n-m) at (90:1.7) {$n\!\!-\!\!m$};

    \foreach \x in {1, 2, dot, n-m}{
        \draw (\x) to node[weight] {$0$} (t);
    }

    \draw (1) to node[weight] {$0$} (2);
    \draw (2) to node[weight] {$0$} (dot);
    \draw (dot) to node[weight] {$0$} (n-m);

    \def\arco{13}
    \def\raio{7}
    \def\bend{50}
    \node[modified large vertice] (n-m+1) at (1.5 * \arco:\raio) {$n\!\!-\!\!m\!\!+\!\!1$};
    \node[modified large vertice, draw=none] (dot2) at (0.5 * \arco:\raio) {$\dots$};
    \node[modified large vertice] (n-1) at (-0.5 * \arco:\raio) {$n\!\!-\!\!1$};
    \node[modified large vertice] (n) at (-1.5 * \arco:\raio) {$n$};

    \foreach \x in {n-m+1, dot2, n-1, n}{
        \draw (\x) to node[weight] {$1$} (t);
    }
    \draw (n-m+1) to node[weight] {$\varepsilon$} (dot2);
    \draw (dot2) to node[weight] {$\varepsilon$} (n-1);
    \draw (n-1) to node[weight] {$\varepsilon$} (n);
    \draw (n-m+1) to[bend left=\bend] node[weight] {$\varepsilon$} (n-1);
    \draw (n-m+1) to[bend left=60] node[weight] {$\varepsilon$} (n);
    \draw (dot2) to[bend left=\bend] node[weight] {$\varepsilon$} (n);


\end{tikzpicture}
\caption[A metric instance of the sequential transportation game which is used to show a lower bound on the Sequential Price of Anarchy for social function $D$]{Graph $G$ where 
$d(u,t) = d(u, v) = 0$ for all $u,v \in \{1,\dots, n-m\}$,
$d(u,t) = 1$ and $d(u, v) = \varepsilon$ for all $u,v \in \{n-m+1, \dots, n\}$.}
\label{fig:pos_U}
\end{figure}

We now show a lower bound on the value of an optimal outcome for $U$.

\begin{proposition}\label{prop_util}
Let $\sigma^*$ be an optimal outcome. If $d$ is metric, then $\sum_{i \in N} d(i,t) \le U(\sigma^*)$.
\end{proposition}
\begin{proof}
Because of the triangle inequality, for a player $i$ we have that ${d(i,t) \le c_i(\sigma^*)}$, and the result follows.
\end{proof}

Next, we show a general bound for $U$ in any outcome of a transportation game. Then, using this result, we show a tight bound on $\SPoA{}(U)$.
 
\begin{lemma}\label{functionU_upperB}
For a metric transportation game with $n$ players, let $\sigma$ be an outcome and $\sigma^*$ an optimal outcome for $U$, then ${U(\sigma) \le (2n-1)U(\sigma^*)}$.
\end{lemma}
\begin{proof}
Let $\sigma$ and $\sigma^*$ be an outcome and optimal outcome for $U$, respectively. Let $p_{i}^{j}$ be the $i$-th player to be picked up by bus $j$ and ${n_j = |\{i: \sigma_i = j \}|}$. Also, we consider that $j_{n_{j}+1} = t$ for all $j \in M$. We have that

\begin{align}
U(\sigma)
& = \sum_{j \in M} \sum_{i \in N \,\vert\, \sigma_i = j} c_i(\sigma) \nonumber\\
& = \sum_{j=1}^{m} \sum_{i=1}^{n_j} i \cdot d(p_{i}^{j},p_{i+1}^{j}) \label{eqx2}\\ 
& \le \sum_{j=1}^{m}\sum_{i=1}^{n_j} i\, (d(p_{i}^{j},t) + d(t,p_{i+1}^{j}))   \label{eqx3}\\ 
&= \sum_{j=1}^{m}\sum_{i=1}^{n_j} (2i-1)d(p_{i}^{j},t)\nonumber \\
&\le (2n-1)\sum_{j=1}^{m}\sum_{i=1}^{n_j} d(p_{i}^{j},t)   \nonumber \\
&\le (2n-1) U(\sigma^*), \label{eqx6}
\end{align}
where Inequality~(\ref{eqx3}) follows from the triangle inequality. Finally, we get the Inequality~\ref{eqx6} from the fact that the summations are composed by $n$ terms of~$d(j_i,t)$, and then we use Proposition~\ref{prop_util} to get the final result.
\end{proof}

\begin{theorem}\label{cor_SPoA_U}
For metric transportation games with $n$ players, ${\SPoA{}(U) = 2n-1}$.
\end{theorem} 
\begin{proof}
The upper bound comes from Lemma~\ref{functionU_upperB}. For the lower bound, consider an instance $(N,M,G)$ where $|N| = |M| =n$, and the graph that is showed in Figure~\ref{fig:SPoA_U}. Let~$\pi_j$, for $j \in M$, be the identity permutation $(1,\ldots, n)$. It is possible to see that there is a strategy profile $\lambda^*$ such that, in its outcome, player~$n$ is on bus $1$ and the remaining players are on bus $2$, and thus $U(\lambda^*)=1$.   

Now, we can show in a similar way done on Theorem~\ref{SPOA_E} that there exists at least one $\SPE{}$ $\lambda$ where, in its outcome, all players choose the same bus since player $n$ can always choose the same bus as player $n-1$ and player $n-1$ can always choose the same bus as player $n-2$, and so on. Therefore we get that~${U(\lambda) = 2n-1}$, which completes the proof.
\end{proof}

\begin{figure}
\centering 
\begin{tikzpicture}    

    \node[large vertice] (t) at (0,0) {$t$};
    \node[large vertice] (1) at (270:1.7) {$1$};
    \node[large vertice] (2) at (270-45:1.7) {$2$};
    \node[large vertice, draw=none] (dot) at (270-90:1.7) {$\dots$};
    \node[large vertice] (n-2) at (90+45:1.7) {$n-2$};
    \node[large vertice] (n-1) at (90:1.7) {$n-1$};

    \foreach \x in {1, 2, dot, n-2, n-1}{
        \draw (\x) to node[weight] {$0$} (t);
    }

    \draw (1) to node[weight] {$0$} (2);
    \draw (n-2) to node[weight] {$0$} (n-1);
    \draw (n-2) to node[weight] {$0$} (dot);
    \draw (2) to node[weight] {$0$} (dot);

    \node[large vertice] (n) at (0:4) {$n$};

    \draw (t) to node[weight] {$1$} (n);

\end{tikzpicture}
\caption[A metric instance of the sequential transportation game which is used to show a lower bound on the Sequential Price of Anarchy for social function $E$]{Graph $G$ where 
$d(u, t) = d(u, v) = 0$ for all $u,v \in \{1, \dots, n-1\}$ and
$d(n,t) = 1$.}
\label{fig:SPoA_U}
\end{figure}
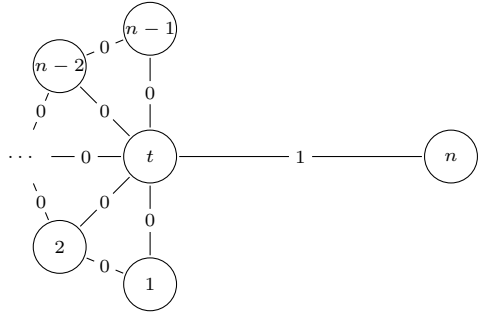

\section{Function $U$ in simultaneous games}\label{sec:U_function}
Since we have introduced the analysis of the utilitarian function $U(\sigma)$, which represents the sum of the distances that players will get to $t$, we will, in this section, compute bounds for both $\POA{}(U)$ and $\POS{}(U)$ for simultaneous transportation games. 

We note that the same instance used by Fotakis et al.~\cite{fotakis2017selfish} to show that the \POA{} is unbounded for functions $D$ and $E$, also shows that the \POA{} is unbounded for function $U$.

\begin{corollary}\label{cor_PoA_U}
For non-metric transportation games, $\POA{}(U)$ is unbounded.
\end{corollary} 

Now, for the metric case, we will analyze its inefficiency by giving bounds on the $\POS{}(U)$ and $\POA{}(U)$. Using the same arguments from Proposition~\ref{prop_util_SPoS}, it is possible to prove the following proposition.

\begin{proposition}\label{prop_util_PoS}
For metric transportation games with $n$ players and $m$ buses, ${\POS{}(U)\ge 2\frac{n}{m}-1}$, even restricted to all buses having the same permutation.\qed
\end{proposition}

We finish this section by proving that the \POA{} is $\Theta(\frac{n}{m})$, where $n$ is the number of players and $m$ is the number of buses. 

\begin{theorem}\label{functionU_upper}
For metric transportation games with $n$ players and $m$ buses, $\POA{}(U) = \Theta(\frac{n}{m})$.
\end{theorem}
\begin{proof}
We first prove the upper bound. Let $\sigma$ be an equilibrium and, as before, let ${n_j = |\{i: \sigma_i = j \}|}$ and $p_{i}^{j}$ be the $i$-th player to be picked up by bus $j$ in $\sigma$, where we consider that $j_{n_{j}+1} = t$ for all $j \in M$. Then, by triangle inequality, we have that, for all buses $j \in M$, 
$$\sum_{k=1}^{n_j} d(p_{k}^{j},p_{k+1}^{j}) \le \sum_{k=1}^{n_j} \left(d(p_{k}^{j},t) + d(t, p_{k+1}^{j})\right) = 2 \sum_{k=1}^{n_j} d(p_{k}^{j},t) - d(p_{1}^{j},t).$$
By summing up for all buses, we have that 
\[
\sum_{j=1}^{m} \left(d(p_{1}^{j},t) + \sum_{k=1}^{n_j} d(p_{k}^{j},p_{k+1}^{j})\right) \le 2 \sum_{j=1}^{m}\sum_{k=1}^{n_j} d(p_{k}^{j},t) = 2\sum_{k=1}^{n}d(k,t).
\]
Therefore, there must exist a bus $j^*$ such that 
\begin{align}
d(p_{1}^{j^*},t) + \sum_{k=1}^{n_{j^*}} d(p_{k}^{j^*},p_{k+1}^{j^*}) \le \frac{2 \sum_{k=1}^{n}d(k,t)}{m}. \label{inq_1}
\end{align}

Since $\sigma$ is an equilibrium, for all players $i$ such that ${\sigma_i \ne j^*}$ and considering that $p^{j^*}_{r}$ is the player that comes after $i$ in $\pi_{j^*}$, we obtain
\begin{align}
    c_i(\sigma) &\le d(i,p^{j^*}_{r}) + \sum_{k=r}^{n_{j^*}} d(p_{k}^{j^*},p_{k+1}^{j^*}) \nonumber\\
                &\le d(i,t)+d(t, p^{j^*}_{r}) + \sum_{k=r}^{n_{j^*}} d(p_{k}^{j^*},p_{k+1}^{j^*}) \label{inq_2} \\
                &\le d(i,t)+ d(t,p_{1}^{j^*})
                + \sum_{k=1}^{r-1} d(p_{k}^{j^*},p_{k+1}^{j^*})
                + \sum_{k=r}^{n_{j^*}} d(p_{k}^{j^*},p_{k+1}^{j^*})
                \label{inq_3}\\
                &= d(i,t)+ d(t,p_{1}^{j^*})
                + \sum_{k=1}^{n_{j^*}} d(p_{k}^{j^*},p_{k+1}^{j^*})
                \nonumber\\
                &\le d(i,t) +  \frac{2 \sum_{k=1}^{n}d(k,t)}{m},\label{inq_4}
\end{align}
where both Inequalities~(\ref{inq_2}) and~(\ref{inq_3}) comes from triangle inequality, and Inequality~(\ref{inq_4}) is obtained using Inequality~(\ref{inq_1}). 

For players $i$ such that ${\sigma_i = j^*}$, we have that
$$c_i(\sigma) \le \sum_{k=1}^{n_{j^*}} d(p_{k}^{j^*},p_{k+1}^{j^*}) \le \frac{2 \sum_{k=1}^{n}d(k,t)}{m}.$$ 

Hence,
\begin{align}
    \sum_{i=1}^{n} c_i(\sigma) &\le \sum_{i=1}^{n} d(i,t) +\frac{2n \sum_{k=1}^{n}d(k,t)}{m} \nonumber\\
                        &\le U(\sigma^*) + \frac{2n}{m}U(\sigma^*)\label{inqx1} \\
                        &\le \left(\frac{2n}{m}+1\right) U(\sigma^*),
\end{align}
where we use Proposition~\ref{prop_util} for Inequality~(\ref{inqx1}).

For the lower bound, we use Proposition~\ref{prop_util_PoS}, which completes the proof.
\end{proof}

\section{Conclusion and Future Work}
In this paper, we have extended the game proposed by Fotakis et al.~\cite{fotakis2017selfish} by considering it in its extensive form, which we call sequential transportation games. As a result of it, we were able to give bounds for the Sequential Price of Stability and the Sequential Price of Anarchy considering social cost functions~$D$, and~$E$, previously introduced by Fotakis et al.~\cite{fotakis2017selfish}, and a new utilitarian social cost function $U$ in simultaneous games. Another contribution is the analyses of the inefficiency of equilibria for function $U$. All bounds presented are asymptotically tights for a constant number of buses, with most of them being tight or asymptotically tight even without this restriction.

As a future direction, we also believe, as Fotakis et al.~\cite{fotakis2017selfish} do, that analyzing the game with different ways of computing the routes of the buses would be of great interest. Also, closing the gaps of the Table~\ref{tab_Results} is another interesting direction for the research.

\section*{Acknowledgments}
This research was partially supported by the S\~ao Paulo Research Foundation (FAPESP) grants 2015/11937-9, 2016/01860-1, and 2017/05223-9; and the National Council for Scientific and Technological Development (CNPq) grants 308689/2017-8, 425340/2016-3, 314366/2018-0, and 425806/2018-9. This study was financed in part by the Coordena\c{c}\~ao de Aperfei\c{c}oamento de Pessoal de N{\'i}vel Superior - Brasil (CAPES) - Finance Code 001.

\bibliography{mybibs}

\end{document}